\documentclass[11pt]{article}
\usepackage{amsmath,amssymb,latexsym}
\usepackage{theorem}

\newtheorem{definition}{Definition}
\newtheorem{lemma}{Lemma}

\newtheorem{corollary}{Corollary}
\newtheorem{theorem}{Theorem}
\newtheorem{proposition}{Proposition}

\newtheorem{open-probl}{Open problem}

\newtheorem{exampleA}{Example}
\newtheorem{remarkA}{Remark}

\newenvironment{proof}[1][Proof.]
    {\addvspace{\bigskipamount}\noindent\emph{#1} }
    {\par\addvspace{\bigskipamount}}

\newcommand{\qed}{\ifmmode\raisebox{3.3pt}{\fbox{}}\else{\ifhmode
\unskip\fi\nobreak\hfil\penalty50\hskip1em\null\nobreak\hfil\raisebox
{3.3pt}{\fbox{}}\parfillskip=0pt\finalhyphendemerits=0\endgraf}\fi}

\title{Quantum finite automata and linear context-free languages: a decidable problem
 \thanks{The third author acknowlegdes the support of fundings ``AST 2009'' of the University of Rome ``La Sapienza''.}}
 \author{Alberto Bertoni \\
        Dipartimento di Scienze dell'Informazione, \\
Universit\`a degli Studi di Milano\\
  Via Comelico 39, 20135 Milano, Italy\\
 \and 
 Christian Choffrut \\
        Laboratoire LIAFA, Universit\'e de Paris 7\\ 
          2, pl. Jussieu, 75251 Paris Cedex 05
   \and 
 Flavio D'Alessandro \\
        Dipartimento di Matematica, \\
``La Sapienza'' Universit\`a di Roma  \\
  Piazzale Aldo Moro 2, 00185 Roma,  Italy.  }

 \date{}
 
 \newcommand{\G}{C}
 \newcommand{\closure}{\mathbf Cl}
 \newcommand{\Closure}[1]{{\mathbf Cl}(#1)}

\newcommand{\N}{\mathbb{N}}

\newcommand{\Q}{\mathbb{Q}}

\newcommand{\R}{\mathbb{R}}

\newcommand{\Hist}{Hist}

\newlength{\transwidth}
\newcommand{\trans}[1]{\settowidth{\transwidth}{$\;_{#1}\;$}
\stackrel{#1}{\overrightarrow{\rule{\transwidth}{0ex}}}}
 
\begin{document}

\maketitle

\begin{abstract}
We consider the so-called measure once finite quantum automata model introduced by Moore and Crutchfield in 2000.
We show that given a language recognized by such a device and a linear context-free language, it is recursively decidable 
whether or not they have a nonempty intersection.
This extends a result of Blondel et al. which can be 
interpreted as solving the problem with the  free monoid in place of the family of linear 
context-free languages.
 
\end{abstract}

%{\bf keywords:} {\em Quantum automata, Context-free languages, Algebraic groups, Decidability}
\section{Introduction}
Quantum finite automata or simply quantum  automata were introduced at the beginning of the
previous decade in \cite{moore-crutchfield} as a new model of language recognizer (cf. \cite{BMP03}). 
Numerous publications have 
ever since compared their decision properties
to those of the older model of probabilistic finite automata. 
Some undecidable problems 
for probabilistic finite automata turn out to be  
decidable for quantum finite automata.
The result in  \cite{blondel} which triggered our investigation can be viewed as asserting that
the intersection emptiness problem  of a language recognized by a finite 
quantum automaton with the free monoid is recursively decidable.
The present result concerns  the same problem 
where instead of the free monoid,  more generally a language belonging to some classical 
families of languages such as the context-free languages and the 
bounded semilinear  languages is considered. 

An ingredient of the proof in  \cite{blondel} consists of expressing
the emptiness problem in the first order theory 
of the reals and then to apply Tarski-Seidenberg quantifier elimination.
This is possible because an algebraic subset, i.e., 
a closed subset in the Zariski topology ${\cal A} \subseteq \R ^ {n}$,
is naturally associated to this intersection and even more miraculously
because this subset can be effectively computed (cf. also \cite{derksen}). 

Here we show that the (actually semi-)algebraicity of ${\cal A}$  still  holds
when considering not only the free monoid but more generally 
arbitrary context-free languages and bounded semilinear languages. 
Unfortunately, its effective construction is only guaranteed under stricter 
conditions such as the fact 
that the language is context-free and linear or is bounded semilinear. In 
particular, in the case of 
context-free languages, we are not able to settle the nonlinear case yet. 
 
 We now give a  more formal presentation of our work. 
The free monoid generated by the finite alphabet
$\Sigma$ is denoted by $\Sigma^*$.
The elements of $\Sigma^*$ are \emph{words}.
We consider all finite dimensional
vector spaces as  provided with the Euclidian norm $||\cdot||$.
A \emph{quantum automaton} is a quadruple ${\cal Q}=(s,\varphi,P,\lambda)$ 
where $s\in \R^{n}$
is a row-vector of unit norm,
$P$ is a projection of $\R^{n}$,  $\varphi$ is a representation of the free monoid 
$ \Sigma^*$  into the group
of \emph{orthogonal} $n\times n$-matrices in  $\R^{n\times n}$
and  the \emph{threshold}  $\lambda $  has value in $\R$.
We recall that a real matrix $M$ is orthogonal
if its inverse equals its transpose: $M^{-1}=M^{T}$.
We denote by $O_{n}$ the group of $n\times n$-orthogonal matrices.
We are mainly interested in effective properties
which require the quantum automaton
to be effectively given. 
We say that the quantum automaton is \emph{rational} if  all the coefficients 
of the components of the automaton 
are rational numbers,  i.e., $\varphi$ maps $ \Sigma^*$  into $\Q^{n\times n}$
and $\lambda\in \Q$. This hypothesis is not a  restriction 
since all we use for the proofs is the fact 
that the arithmetic operations and the comparison
are effective in the field of rational numbers.
This is the {``measure once''} model introduced 
 by Moore and Crutchfield in 2000  \cite{moore-crutchfield}.
%We define the language recognized by ${\cal Q}$ as
%%
%\begin{equation}
%\label{eq:Q>}
%|{\cal Q}_{>}|=\{w\in  \Sigma^*\mid ||s \varphi(w) P||>\lambda\}
%\end{equation}
%Blondel et al. in \cite{blondel} proved that the emptiness problem of 
%$|{\cal Q}_{>}|$ is decidable. This can be interpreted as saying that 
%the emptiness problem of the intersection 
%of a language accepted by a quantum automaton and 
%the specific language $\Sigma^*$
%is decidable. In other word, it falls into the category 
%of issues asking for the decision status of the intersection
%of two languages. It is known that such a problem 
%is already undecidable at a very low level of the complexity
%hierarchy of recursive 
%languages, namely for linear context-free languages
%to which  Post Correspondence Problem can be easily reduced.  It is worth to remark that, in \cite{blondel}, the emptiness problem of 
%$|{\cal Q}_{\geq}|$ has been proved to be undecidable, with the natural 
%meaning of the notation $|{\cal Q}_{\geq}|$.
% 
For a real threshold $\lambda$, the languages recognized by  $\cal Q$ with strict and nonstrict threshold $\lambda$ are
%\begin{equation}
%\label{eq:Q>}
$$|{\cal Q}_{>}|=\{w\in  \Sigma^*\mid ||s \varphi(w) P||>\lambda\}, \quad |{\cal Q}_{\geq}|=\{w\in  \Sigma^*\mid ||s \varphi(w) P||\geq\lambda\}$$
%\end{equation}
Blondel et al. in \cite{blondel} proved that the emptiness problem of 
$|{\cal Q}_{>}|$ is decidable and the emptiness problem of 
$|{\cal Q}_{\geq}|$ is undecidable. 
It is worth to remark that these results  are in contrast with the corresponding ones  for probabilistic finite automata. Indeed for this class of automata % where 
the above mentioned  problems   
are  both undecidable (see \cite{Paz}, Thm 6.17). 
It is also proven in \cite{j02} that the problems remain both undecidable  for 
  the {``measure many''} model of quantum automata, a model not computationally equivalent to the {``measure once''},  introduced 
 by Kondacs and Watrous in \cite{kw}.
  
The result of decidability proved in \cite{blondel} for quantum automata can be interpreted as saying that 
the emptiness problem of the intersection 
of a language accepted by a quantum automaton and 
the specific language $\Sigma^*$
is decidable. 
In other word, it falls into the category 
of issues asking for the decision status of the intersection
of two languages. It is known that such a problem 
is already undecidable at a very low level of the complexity
hierarchy of recursive 
languages, namely for linear context-free languages
to which  Post Correspondence Problem can be easily reduced.  
 
  A few words on the technique 
used in the above paper. Observe that, with the natural 
meaning of the notation $|{\cal Q}_{\leq}|$, the emptiness problem for languages  $|{\cal Q}_{>}|$
  is equivalent to the inclusion 
\begin{equation}
\label{eq:l}
 \Sigma^* \subseteq |{\cal Q}_{\leq}|
\end{equation}
Since the function $M \rightarrow ||s M P||$
is continuous, it is sufficient to prove that for all
matrices $M$ in the topological closure of 
$\varphi(\Sigma^*)$ the condition 
$||s M P||\leq \lambda$ holds.
The  nonemptiness is clearly semidecidable.
In order to prove that the emptiness is semidecidable 
the authors resort to two ingredients. 
They observe that the topological closure of the
monoid of matrices $\varphi(\Sigma^*)$
is algebraic, i.e., when considering 
the $n\times n$-entries of a matrix $M$ in the topological closure of $\varphi (\Sigma^*)$
as as many unknowns in the field 
of reals, they are precisely the zeros 
of a polynomial in $\R[x_{1,1}, \ldots, x_{n,n}]$.
This allows them to express the property (\ref{eq:l})  
in first-order logic of the field of reals.
The second ingredient consists of applying 
Tarski-Seidenberg quantifier elimination
 and Hilbert basis results,  which yields decidability.

   We generalize the problem by considering 
families of languages ${\cal L}$ instead of the fixed language
$\Sigma^*$. The  question we tackle is thus the following

\medskip

\noindent
{\sc $(L, Q)$ INTERSECTION}

\medskip

\noindent {\sc Input}:  a  language $L$ in a family of
languages ${\cal L}$ and a finite quantum automaton
${\cal Q}$. 

\medskip

\noindent {\sc Question}:  does $L\cap |{\cal Q}_{>}|=\emptyset$ hold?
 
\medskip Our main result shows that whenever 
${\cal L}$ is the family of linear context-free languages or
  is the family of bounded semilinear   languages,
and whenever the automaton is rational, the problem is decidable.
It can be achieved, not only because
the orthogonal matrices associated with 
$L$ are semialgebraic (a more general property than
algebraic,  which is defined by more general
first-order formulas), but also because these formulas 
can be  computed  ``in the limit''. 

We can prove the semialgebraicity of more general families 
of languages:  arbitrary subsemigroups
which is a trivial case and context-free languages
 which is less immediate.
 Finally we show that our main results are not 
 trivial since we can   exhibit
an example of a language which is the complement
of a    context-free language and 
whose set of matrices is  not semialgebraic. 

Some of the results of this paper will be presented at DLT 2013 \cite{dlt2013}.

\section{Preliminaries}
A quantum automaton  ${\cal Q}$ is a quadruple  $(s,\varphi, P,\lambda)$
where, as mentioned in the Introduction,  $s\in \R^{n}$
is a vector of unit norm,
$P$ is a projection of $\R^{n}$,  
$\varphi$ is a representation of the free monoid 
$ \Sigma^*$  into the group $O_n$
of orthogonal $n\times n$-matrices in  $\R^{n\times n}$.
The behaviour of ${\cal Q}$  heavily depends 
on the topological properties of the semigroup of matrices
$\varphi(\Sigma^*)$. This is why, before returning
to quantum automata, we first focus our attention
on these matrices for their own sake.

\subsection{Topology}

The following result is needed in the proof of the main theorem.
Though valid under weaker conditions,
it will be considered in the particular case of orthogonal
matrices. Given a subset $E$ of a finite dimensional vector space, we denote by 
$\closure(E)$ the topological closure 
for the topology induced by the Euclidian norm.
Given a $k$-tuple of matrices $(M_{1}, \ldots, M_{k})$, 
 denote by $f$  the $k$-ary product 
 $f(M_{1}, \ldots, M_{k})=M_{1} \cdots M_{k}$
 and extend the notation to subsets  $\rho$ 
 of $k$-tuples of matrices by posing 
 $f(\rho)=\{f(M_{1}, \ldots, M_{k})\mid (M_{1}, \ldots, M_{k}) \in \rho\}$.
The following result will be applied in several instances of this paper.
It says that because we are dealing with compact subsets,  the two operators
of matrix multiplication and  the topological
closure   commute. % (For the proof of Theorem \ref{th:fundamental}, see the Appendix).
\begin{theorem}
\label{th:fundamental} 
Let $\cal C$ be  a compact subset of matrices  and let $\rho \subseteq {\cal C}^k$ 
be a $k$-ary relation. Then we have
$$
\Closure{ f(\rho)}= 
f( \Closure{\rho})
$$
\end{theorem}
  
\begin{proof} Since the function $f$ is continuous, 
the inverse image of $\Closure{f(\rho)}$ is closed,
i.,e.,  $\Closure{\rho}\subseteq f^{-1}(\Closure{f(\rho)})$
holds which yields $f(\Closure{\rho})\subseteq \Closure{f(\rho)}$. 
Now we prove the opposite inclusion. Consider
an element $A\in \Closure{f(\rho)}$. It 
is the limit of a sequence $M_{1,n}\cdots M_{k,n}$
where $(M_{1,n}, \cdots,  M_{k,n})\in \rho$ for $n\geq 0$. 
Because $\cal C$ is a compact set, 
there exists a subsequence $(M_{1,n_{i}}, \cdots,  M_{k,n_{i}})\in \rho$, i.e., 
an infinite sequence of strictly increasing indices $n_{i}$ 
which converges to a limit point $(A_{1},\cdots, A_{k})\in \Closure{\rho}$. By continuity
we have $f(A_{1},\cdots, A_{k})=A$ which shows that $\Closure{f(\rho)}\subseteq f(\Closure{\rho})$. 
\qed\end{proof}

Consequently, if $\rho$ is a binary relation which is a direct product $\rho_{1}\times \rho_{2}$, 
we have $\Closure{\rho_{1} \rho_{2}}=f(\Closure{\rho_{1}\times \rho_{2}})$.
It is an elementary 
result of topology  that 
$\Closure{\rho_{1}\times \rho_{2}}=\Closure{\rho_{1} }\times\Closure{\rho_{2} }$ holds.
Because of 
$\Closure{\rho_{1} \rho_{2}}=f(\Closure{\rho_{1}\times \rho_{2}})
=f(\Closure{\rho_{1}}\times \Closure{\rho_{2}})
=\Closure{\rho_{1} }\  \Closure{\rho_{2} }$
we have
\begin{corollary}
\label{cor:closure-of-product}
The topological closure of the product of two sets of matrices 
included in a compact subspace is equal to the product of the topological
closures of the two sets.
 \end{corollary}

 \subsection{Algebraic and semialgebraic sets}
Let us give first the definition of algebraic 
set over the field of real numbers (cf. \cite{BasuetAl,oni}). 
 \begin{definition}
A subset ${\cal A} \subseteq \R ^ {n}$ is   
{\em algebraic (over the field of real numbers)},  if it satisfies one of the following
equivalent conditions:

\medskip

\noindent  {(i)}
${\cal A}$ is the zero set of a polynomial
$p\in \R[x_{1}, \ldots, x_{n}]$, i.e.,
  \begin{equation} 
  v \in {\cal A} \ \Longleftrightarrow \   p(v)=0.
  \end{equation}

\medskip

\noindent  {(ii)}
${\cal A}$ is the zero set of an arbitrary set of polynomials
 $\cal P$  with coefficients 
in $\R[x_{1}, \ldots, x_{n}]$, i.e.,
for every vector  $v\in  \R ^ {n}$,
  \begin{equation}
  \label{eq:infinite-polynomials}
  v \in {\cal A} \ \Longleftrightarrow \ \forall \ p\in {\cal P}: \ p(v)=0.
  \end{equation}

\end{definition}

The equivalence of the two statements is a consequence
of Hilbert finite basis Theorem. Indeed, the theorem claims that given a 
family  ${\cal P}$ there exists a finite
subfamily $p_{1}, \ldots, p_{r}$ generating the same 
ideal which implies in particular that for all $p\in {\cal P}$ there exist
$q_{1}, \ldots, q_{r}$ with 
$$
 p= q_{1}p_{1} + \cdots + q_{r}p_{r}
$$
Then $p_{j}(v)=0$ for $j=1, \ldots, r$ implies
$p(v)=0$. Now this finite set of equations can be reduced
to the single equation
 $$\displaystyle \sum^n_{i=1} p_{j}(x)^2=0
 $$
As a trivial example, a singleton $\{v\}$ is algebraic since it is the unique
solution of the equation $$\displaystyle \sum^n_{i=1} (x_{i}-v_{i})^2=0$$
where $v_i,$ with $1\leq i\leq n,$ is the i-th component of the vector $v$.

\bigskip It is routine to check that
the family of  algebraic sets  is closed under finite unions
and intersections. However, it is not closed under  complement and projection.
The following more general class of subsets enjoys 
extra closure properties and is therefore more robust.
The equivalence of the two definitions
below is guaranteed by  Tarski-Seidenberg  quantifier elimination result.

 \begin{definition}
A subset ${\cal A} \subseteq \R ^ {n}$ is   
 {\em semialgebraic (over the field of real numbers)} if it satisfies one of the two
 equivalent conditions
 
 \medskip

\noindent {(i)} ${\cal A}$ is the set of vectors satisfying 
a finite Boolean
combination of predicates of the form $p(x_{1}, \ldots, x_{n})>0$ 
where $p\in \R[x_{1}, \ldots, x_{n}]$.

 \medskip

\noindent {(ii)} ${\cal A}$ is first-order definable in the theory of
 the structure whose domain are the reals 
and whose predicates are of the form 
$p(x_{1}, \ldots, x_{n})>0$ and $p(x_{1}, \ldots, x_{n})=0$ 
with $p\in \R[x_{1}, \ldots, x_{n}]$.
 
\end{definition}
 
  We now specify these definitions to square matrices. 

 \begin{definition}
 A set  ${\cal A} \subseteq \R ^ {n\times n}$  of  matrices 
is {\em algebraic}, resp.  {\em semialgebraic}, if considered as a
set of vectors of dimension $n^2$, it is algebraic,  resp.   semialgebraic.
\end{definition}

  We now combine the notions
of zero sets and of topology.
In the following two results  we rephrase Theorem 3.1 of
 \cite{blondel} by emphasizing the main features that serve our purpose (see also  \cite{blondel,oni}).  Given a
subset $E$ of a group, we denote by $\langle E \rangle$ and by $E^*$ 
the subgroup and the submonoid it generates, respectively.
% (for the proof of Theorem \ref{th:semigroup-in-compact} and  Theorem \ref{th:topological-closure-of-monoid}, see the Appendix).
 
\begin{theorem}
\label{th:semigroup-in-compact}
Let $S\subseteq \R^{n\times n}$ be a set of orthogonal  matrices
and let $E$ be any subset  of $S$ satisfying
$\langle S \rangle = \langle E \rangle$.
Then  we have $\Closure{S^*}=\Closure{\langle E \rangle}$. In particular $\Closure{S^*}$ is a group. 
\end{theorem}
\begin{proof} 
 It is known that every compact subsemigroup of a
compact group is a subgroup $G$.
Now $S^*\subseteq \langle E \rangle$
implies $G=\Closure{S^*}\subseteq \Closure{\langle E \rangle}$
and $S\subseteq G$ implies
$\Closure{\langle E \rangle}\subseteq G$
and thus  $\Closure{S^*}= \Closure{\langle E \rangle}$.
 \qed\end{proof}

The main consequence of the next theorem
is that the  topological closure 
of a monoid of orthogonal matrices 
is algebraic

\begin{theorem}
\label{th:topological-closure-of-monoid}
Let $E$ be a set of orthogonal  matrices. 
Then   $\Closure{\langle E \rangle}$
is a subgroup of orthogonal matrices
and it is the zero set of 
all polynomials $p[x_{1,1},\ldots, x_{n,n} ]$
satisfying the conditions
$$p(I)=0 \  \mbox{ and } p(eX)=p(X)  \  \mbox{ for all  } e\in  E$$
Furthermore, if the matrices in $E$ have rational coefficients,
the above condition may be restricted to polynomials with
coefficients in $\Q$.

\end{theorem}

\begin{proof} It is clear that $\Closure{\langle E \rangle}$ is a 
subgroup of orthogonal matrices, 
say $G$. By \cite[Thm 5, p. 133]{oni}
this group is  the zero set of all polynomials  
$p[x_{1,1},\ldots, x_{n,n} ]$
satisfying the conditions (where $I$ denotes the identity matrix)
\begin{equation}
\label{eq:density}
p(I)=0 \  \mbox{ and } p(gX)=p(X)  \  \mbox{ for all  } g\in  G
\end{equation}
Let us verify that
we may assume the above condition is satisfied by
all $e\in E$. First, if it is the case, it is satisfied for
all elements of the group $\langle E \rangle$. Now observe that
condition (\ref{eq:density}) defines a linear constraint on the 
coefficients of the polynomial. Indeed if $V\in \R^{d}$ 
is the vector of 
coefficients of the polynomial $p$ and if $X$ is viewed as a set of $n^2$ variables, 
identifying the coefficients of the monomials  of the two hand-sides of (\ref{eq:density}) yields   a system of linear equations
$$
M V=V
$$
where the matrix $M$ depends on
$g$ only, say $M=M_{g}$.
Let 
$$
\lim_{i\rightarrow \infty} g_{i}= g
\ \mbox{ and } M_{g_{i}} V=V \ \mbox{ for all  } i\geq 0
$$
Then by continuity we have $M_{g}V=V$. 

The last assertion concerning the case where the coefficients
are rational can be found in \cite[Th. 3.1]{blondel}.
\qed\end{proof}

Combining the previous two 
theorems, we get the general result

\begin{corollary}
\label{cor:L-star-is-algebraic}
Let $L\subseteq \Sigma^*$. Then $\Closure{\varphi(L)^*}$
is algebraic.
\end{corollary}

\subsection{Effectiveness issues}

 We now return to the $(L, Q)$ INTERSECTION problem
as defined in the Introduction. We want to prove the 
implication
$$
\forall X: X \in \varphi(L) \Rightarrow ||sXP||\leq \lambda
$$
We observed that due to the fact that the function
$X\rightarrow ||sXP||$ is continuous the implication is
equivalent to the implication
$$
\forall X: X \in \Closure{\varphi(L)} \Rightarrow ||sXP||\leq \lambda
$$
It just happens that under certain hypotheses, 
$\Closure{\varphi(L)}$ is semialgebraic, i.e.,
it is defined by a first-order formula 
which turns the above statement into a first order
formula. In the simplest examples, the closure is defined by an infinite 
conjunction of equations which by  Hilbert  finite basis result 
reduces to a unique equation. Thus
Theorem \ref{th:topological-closure-of-monoid} guarantees the existence
of the formula but does not give an upper bound on
the finite number of equations which must be tested.
Therefore the following definition is instrumental 
for the rest of the paper.
It conveys the idea that given a subset $\cal A$ of matrices  there exists a sequence of formulas 
defining a non-increasing sequence of matrices which eventually coincide with $\cal A$. 
Each formula of the sequence can thus be considered as an approximation of the ultimate formula.

 \begin{definition}
\label{de:effectively-eventually-definable}
A subset ${\cal A}$ of  matrices 
is \emph{effectively eventually definable} if there exists
a constructible sequence of first-order formulas $\phi_{i}$
satisfying the conditions

\medskip

\noindent 1)   for all $i\geq 0$ $ \phi_{i+1}\Rightarrow \phi_{i}$

\medskip

\noindent  2)  for all $i\geq 0$ ${\cal A} \models \phi_{i}$

\medskip

\noindent  3) there exists  $n\geq 0 $  ${\cal B} \models \phi_{n} \Rightarrow {\cal B} \subseteq {\cal A}$ 
%
%
%\nl 1)  for all $i\geq 0$, for all $X\in \R^{n\times n}$
%we have  $ \phi_{i+1}(X)\Rightarrow \phi_{i}(X)$
%
%\nl  2) for all $i\geq 0$, for all $X\in {\cal A}$ we have $ \phi_{i}(X)$
%
%
%\nl  3) there exists an integer $n\geq 0 $ such 
%that 
%%
%$$
%  \phi_{n}(X)  \Rightarrow X\in {\cal A}
%$$
\end{definition}
   The following is a first application of the notion
and illustrates the discussion before the definition. % (for the proof of the next two propositions, see the Appendix).

 \begin{proposition}
\label{pr:meta-resuositionlt}
Let ${\cal Q}$   be a rational quantum automaton. Let $L\subseteq \Sigma^*$ 
be such that   the set $\Closure{\varphi(L)}$ is effectively
eventually definable. It is recursively decidable 
whether or not  
$L\cap |{\cal Q}_{>}|=\emptyset$ holds.
\end{proposition}

\begin{proof}
Equivalently we prove the inclusion $L\subseteq |{\cal Q}_{\leq}|$. 
In order to prove that the 
inclusion is effective, we proceed as in \cite{blondel}.
We run in parallel two semialgorithms.
The first one verifies the noninclusion by enumerating the words 
$w\in L$ and testing if 
$||s \varphi(w) P||>\lambda$ holds.
The second semialgorithm considers
a sequence of formulas $\phi_{i}(X)$, $i=0,\ldots,$
which effectively eventually defines 
$\Closure{\varphi(L)}$ and verifies whether
the sentence
$$
\Psi_{i}\equiv \forall X: \phi_{i}(X) \Rightarrow
sXP\leq \lambda
$$
holds which can be achieved by Tarski Seidenberg elimination 
result.
If the inclusion $L\subseteq |{\cal Q}_{\leq}|$ 
holds then the first
semialgorithm cannot answer ``yes'' and the 
second semialgorithm will eventually
answer ``yes''. If the inclusion does not
hold then the second semialgorithm 
cannot answer ``yes'' for any $\Psi_{i}$
 since the second condition
of the Definition \ref{de:effectively-eventually-definable} implies
$X\in \Closure{\varphi(L)} \Rightarrow \phi_{i}(X) $
and thus
$$\forall X: X\in \Closure{\varphi(L)} \Rightarrow  ||sXP|| \leq \lambda $$
a contradiction.
\qed\end{proof}

We state a sufficient condition for a subset of matrices
to be effectively
eventually definable.

Let $S\subseteq \R^{n\times n}$ be a set of orthogonal  matrices
and let $E$ be any subset satisfying
$\langle S \rangle = \langle E \rangle$.

\begin{proposition}
\label{pr:finitely-generated}
Let $L\subseteq \Sigma^*$ and let $E\subseteq \Q^{n\times n}$ 
be a finite subset of orthogonal matrices 
 satisfying $\langle\varphi(L)\rangle = \langle E\rangle$. 
Then $\Closure{\varphi(L)^*}$ is effectively
eventually definable.
 \end{proposition}
 
 \begin{proof}
Indeed, set ${\cal A}=\Closure{\varphi(L)^*}=\Closure{\langle E\rangle}$
where the last equality is guaranteed by Theorem \ref{th:semigroup-in-compact}.
Then  ${\cal A}$  is the zero set  
of all polynomials $p(X)$ where 
$p$ satisfies the condition
$$
p(I)=0 \mbox{ and } p(gX)=p(X)  \mbox{ for all } g\in {\cal A}
$$
Since it clearly 
suffices to verify the invariance of $p$ under the action
of the finite set of generators, we proceed as follows.
We enumerate all polynomials $p\in \Q[x_{1,1}, \ldots, x_{n,n}]$
say $p_{0}, p_{1}, \ldots $.
For each such polynomial $p$ the invariance 
relative to the action of each generator can be tested.
Thus the formula
$$
\phi_{i}(X)\equiv \bigwedge^{i}_{j=1}  p_{j}(X)=0
$$
effectively eventually defines ${\cal A}$: the first
two conditions can be readily verified and the last one is 
a consequence of Hilbert finite basis theorem on ideals
of polynomials.
\qed\end{proof}

%\label{pr:finitely-generated} 
%Indeed, set ${\cal A}=\Closure{\varphi(L)^*}=\Closure{\langle E\rangle}$
%where the last equality is guaranteed by Theorem \ref{th:semigroup-in-compact}.
%Then  ${\cal A}$  is the zero set  
%of all polynomials $p(X)$ where 
%$p$ satisfies the condition
%%
%$$
%p(I)=0 \mbox{ and } p(gX)=p(X)  \mbox{ for all } g\in {\cal A}
%$$
%%
%Since it clearly 
%suffices to verify the invariance of $p$ under the action
%of the finite set of generators, we proceed as follows.
%We enumerate all polynomials $p\in \Q[x_{1,1}, \ldots, x_{n,n}]$
%say $p_{0}, p_{1}, \ldots $.
%For each such polynomial $p$ the invariance 
%relative to the action of each generator can be tested.
%Thus the formula
%%
%$$
%\phi_{i}(X)\equiv \bigwedge^{i}_{j}  p_{j}(X)=0
%$$
%%
%effectively eventually defines ${\cal A}$: the first
%two conditions can be readily verified and the last one is 
%a consequence of Hilbert finite basis theorem on ideals
%of polynomials.
% \end{proof}

\subsection{Closure properties}
In this paragraph we investigate some closure properties of
the three different classes of matrices:
algebraic, semialgebraic and effectively eventually definable,
under the main usual operations as well as new operations.

  We define the  \emph{sandwich}  operation denoted by  $\diamond$
whose first operand is a set of pairs of matrices
 ${\cal A}\subseteq \R^{n\times n} \times \R^{n\times n} $ 
and the second operand a set of matrices ${\cal B}\subseteq \R^{n\times n} $
by setting  
$$
{\cal A} \diamond {\cal B}= \{XYZ\mid (X,Z)\in {\cal A} \mbox{ and } Y\in {\cal B}\}
$$
The next operation will be used. Given a  bijection 
\begin{equation}
\label{eq:entry-renaming}
\pi: \{(i,j) \mid  i, j\in \{1, \ldots, n\}\}\rightarrow \{(i,j) \mid  i, j\in \{1, \ldots, n\}\}
\end{equation}
and a matrix $M \in   \R ^ {n\times n}$ denote by $\pi(M)$
the matrix $\pi(M)_{i,j}=M_{\pi(i, j)}$. Extend this operation to subsets 
of matrices ${\cal A}$. Write $\pi({\cal A})$ to denote the set
of matrices $\pi(M)$ for all $M \in   {\cal A}$.

 The last operation is the
\emph{sum} of square  matrices  $M_{1},  \ldots ,  M_{k}$
whose result  is the square block matrix 
\begin{equation}
\label{eq:sum-of-matrices}
M_{1}\oplus \cdots \oplus M_{k}=              \left (\begin{array}{cccc} 
            M_{1} &   0 &  \cdots &   0\\ 
            0 &   M_{2}   & \cdots &  0\\
            \vdots     & \vdots & \vdots & \vdots \\ 
             0 &   0  &   0  & M_{k}
            \end{array}
\right )
\end{equation}
These notations extend to subsets of matrices in the 
natural way.
Here we assume that all $k$ matrices have the same dimension $n\times n$. 
Observe that if the matrices are orthogonal,
so is their sum. Such matrices form a subgroup of
orthogonal matrices of dimension $kn\times kn$.

 Logic provides an elegant way to formulate properties
in the present context.  
  Some conventions are used throughout this work. E.g.,
we write $\exists^n X$ when we mean that $X$
is a vector of $n$ bound variables. Furthermore, 
a vector of $n\times n$ variables can be interpreted 
as an $n\times n$ matrix of variables.
 As a consequence of  Tarski-Seidenberg  result,
consider two semialgebraic subsets of matrices, 
say ${\cal A}_{1}$  and ${\cal A}_{2}$, defined 
by two first-order formulas $\phi_{1}(X_1)$  and $\phi_{2}(X_2)$
where $X_1$ and $X_2$ are two families 
of $n^2$ free variables viewed as two
$n\times n$ matrices of variables. Then 
the product 
$$
{\cal A}_{1} {\cal A}_{2}=\{M_{1}M_{2} \mid  M_{1}\in {\cal A}_{1}, M_{2}\in  {\cal A}_{2}\}
$$
is defined by the following formula where $X$ is a family
of $n^2$ free variables viewed as  an
$n\times n$ matrix
$$
\exists^{n\times n} X_{1}\exists^{n\times n} X_{2}: X=X_{1} X_{2}\wedge \phi_{1}(X_{1}) \wedge \phi_{2}(X_{2})
$$
where $X=X_{1} X_{2}$ is an abbreviation for 
the predicate defining $X$ as the matrix product 
of $X_{1}$ and $ X_{2}$. This proves that the product 
of two semialgebraic sets of  matrices is semialgebraic. 
Similarly we have the following closure properties whose
verification is routine.

\begin{proposition}
\label{pr:union-prod-phi} 
Let $ {\cal A}_{1},   {\cal A}_{2}\subseteq   \R ^ {n\times n}$ be two 
 sets of matrices and 
let $\pi$ be a one-to-one mapping as in (\ref{eq:entry-renaming}). 

\medskip

\noindent 1) If $ {\cal A}_{1}$ and $ {\cal A}_{2}$  are algebraic so are 
$ {\cal A}_{1}  \cup {\cal A}_{2}$ and $\pi( {\cal A}_{1})$.

\medskip

\noindent 2) If $ {\cal A}_{1}$ and $ {\cal A}_{2}$  are semialgebraic,
resp. effectively eventually definable,  so are 
$ {\cal A}_{1} \cup  {\cal A}_{2}$,  $ {\cal A}_{1}  {\cal A}_{2}$ and $\pi( {\cal A}_{1})$.

\end{proposition}

\begin{proposition}
\label{pr:sandwich}  Let $ {\cal A}_{1}\subseteq  \R ^ {n\times n}\times \R ^ {n\times n}$   
and $   {\cal A}_{2}\subseteq   \R ^ {n\times n}$ be semialgebraic,
resp. effectively eventually definable.
Then $ {\cal A}_{1} \diamond   {\cal A}_{2}$ is semialgebraic,
resp. effectively eventually definable.
\end{proposition}

 \begin{proposition}
\label{pr:product-of-blocks} 
Let $ {\cal A}$ be a semialgebraic,
resp. effectively eventually definable, set
of $kn\times kn$  matrices of the form (\ref{eq:sum-of-matrices}).
The set 
$$
\{X_{1}\cdots X_{k}\mid X_{1}\oplus \cdots \oplus X_{k} \in {\cal A}\}
$$
is semialgebraic,
resp. effectively eventually definable.
\end{proposition}

 \begin{proposition}
\label{pr:blocks-of-product} 
If $ {\cal A}_{1},   \ldots, {\cal A}_{k}\subseteq   \R ^ {n\times n}$ are  
semialgebraic,
resp. effectively eventually definable sets of matrices then so is the set
$ {\cal A}_{1}\oplus   \cdots \oplus{\cal A}_{k}$.

\end{proposition}

\section{Context-free languages}
\label{sec:context-free}

For the sake of self-containment 
and in order to fix notation, we recall the basic properties 
and notions concerning the family
of context-free languages which can be found in all 
introductory textbooks on theoretical computer science (see, for instance, \cite{Hu}).

A  \emph{context-free grammar} $G$ is a quadruple $\langle V, \Sigma, P, S \rangle$ 
 where $\Sigma$  is the  alphabet of   \emph{terminal symbols},
$V$ is the set of \emph{nonterminal symbols}, $P$ is the set of 
\emph{rules}, and $S$ is the \emph{axiom} of the grammar.  
A word over the alphabet $\Sigma$ is called {\em terminal}.
As usual, the nonterminal symbols are denoted by uppercase 
letters $A$, $B$, \ldots. A typical rule of the grammar is written
as $A\rightarrow \alpha$. 
 The 
\emph{derivation} relation of $G$ is denoted by $ \displaystyle\mathop{\Rightarrow}^{*}$.

A grammar is \emph{linear} if 
every right hand side $\alpha$ contains at most one occurrence of nonterminal symbols,
i.e., if it belongs to $\Sigma^{*}  \cup \Sigma^{*} V \Sigma^{*}$.

The idea of the following
notation is to consider the set of all pairs of 
left and right contexts in the terminal alphabet of a self-embedding nonterminal symbol.
In the next definition, the initial 
``$C$'' is meant to suggest the term ``context''
as justified by the following.

\begin{definition}
\label{de:contexts}
With each nonterminal symbol $A\in V$ associate its
\emph{terminal contexts} defined as  
$$
\G_A = \{(\alpha, \beta) \in \Sigma ^* 
\times \Sigma ^*: A \displaystyle\mathop{\Rightarrow}^{*} \alpha A \beta \}.
$$
\end{definition}

\noindent
It is convenient to define the sandwich operation also for  languages in the following way.  With $C_A$ as above and  $L'$ an arbitrary language, 
 we define% the  \emph{sandwich}  of $C$ and $N$ as the language  %, the set, still denoted by  $ C\diamond L_1$
 $$
C_A \diamond L' = \{uwv\mid (u,v)\in {C_A} \mbox{ and } w\in {L'}\}
$$
As the proof of the main theorem proceeds by induction
on the number of nonterminal symbols, we need to show how to 
recombine a grammar from simpler ones
obtained by choosing an arbitrary non-axiom symbol
as the new axiom and by canceling all the rules 
involving $S$. This is the reason for introducing
the next notation

 \begin{definition}
 \label{de:v-prime}
 Let $G = \langle V, \Sigma, P, S \rangle$ be a  context-free grammar. 
Set $V' = V \setminus \{S\}$.
 
For every $A\in V'$, define the context-free      
grammar $G_{A} = \langle V', \Sigma, P_{A}, A \rangle$ where
the set $P_{A}$  consists of all the rules $B \rightarrow \gamma$ of $G$  of the form
$$
\quad B \in V', \quad \gamma \in (V' \cup \Sigma)^*
$$
and denote by $L_A$ the language of all terminal words 
generated by the grammar $G_A$.
 \end{definition}
The next definition introduces the language of terminal words
obtained in a derivation where   
$S$ occurs at the start only.

 \begin{definition}
 \label{de:L-prime}
Let  $L'(G)$  denote the set of all the words 
of $\Sigma^*$ which admit a derivation
   \begin{equation}
   \label{eq:derivation}
   S \Rightarrow \gamma_{1} \Rightarrow \cdots  \Rightarrow \gamma_\ell  \Rightarrow w
   \end{equation}
where, for every $i=1, \ldots, \ell$,   $\gamma _i\in (V' \cup \Sigma)^*$.

%If no ambiguity arises, in the sequel, the language  $L'(G)$  is simply denoted $L'$.
 \end{definition}

The language $L'(G)$ can be easily expressed in terms 
of the languages $L_{A}$ for all $A\in V'$. Indeed,
consider the set of all rules of the grammar $G$ of the form
 \begin{equation}
  \label{eq:beta}
S\rightarrow \beta, \quad \beta \in (V' \cup \Sigma)^*
\end{equation}

Factorize every such $\beta$ as
 \begin{equation} 
 \label{eqfla:rob}
 \beta =w_{1}A_{1}w_{2}A_{2} \cdots w_{\ell}A_{\ell} w_{j_{\ell+1}}
 \end{equation}
where
$w_{1},  \ldots, w_{\ell+1}\in \Sigma^*$ and $A_{1}, A_{2}, 
\ldots A_{\ell} \in V'$.
The  following is a standard exercise.

\begin{lemma}
\label{le:L'}
 With the notation of (\ref{eqfla:rob}), 
the language $L'(G)$ is  the (finite)
 union of the languages 
$$
w_{1}L_{A_{1}}w_{2}L_{A_{2}} \cdots w_{\ell}L_{A_{\ell}} w_{j_{\ell+1}}
$$
when $\beta$ ranges over all rules  (\ref{eq:beta}). 
 \end{lemma}

\begin{proposition}\label{pr:cfl-decomposition}
With the previous notation $L$ is a finite
union of languages of the form $C_{S}\diamond L''$ where
$$
L''= 
w_{1}L_{A_{1}}w_{2}L_{A_{2}} \cdots w_{\ell}L_{A_{\ell}} w_{\ell+1}
$$
\end{proposition}
 
\begin{proof}
In order to prove the inclusion of the right- 
into left- hand side,  it suffices
to consider $w=  \alpha u \beta,$ with $u\in L'(G) $ and $( \alpha, \beta)\in \G_S.$
One has $S  \displaystyle\mathop{\Rightarrow}^{*}u$ and 
 $S  \displaystyle\mathop{\Rightarrow}^{*} \alpha S \beta$
 and thus  $S  \displaystyle\mathop{\Rightarrow}^{*} \alpha S \beta \mathop{\Rightarrow}^{*}  \alpha u \beta.$

Let us prove the opposite inclusion. A word  $w\in L$ admits a derivation
 $S   \displaystyle\mathop{\Rightarrow}^{*} w $. 
 If the symbol $S$ does not occur in the derivation except 
 at the start of the derivation,  then   $w\in L'(G)$. Otherwise
 factor this derivation into  
 $S   \displaystyle\mathop{\Rightarrow}^{*} \alpha S \beta  \mathop{\Rightarrow}^{*} w $ 
such that $S$ does not occur in the second part of the derivation
except in the sentential form $\alpha S \beta$.
Reorder the derivation $\displaystyle\alpha S \beta  \mathop{\Rightarrow}^{*} w $
into $\displaystyle\alpha S \beta \mathop{\Rightarrow}^{*} \gamma S \delta  \mathop{\Rightarrow}^{*} w $
so that $\gamma,  \delta\in \Sigma^*$.
This implies  $w=\gamma u \delta$ for some word $u\in L'(G)$, %  where $L'(G)$% is defined as in Definition \ref{de:L-prime}, 
 completing the 
proof.	\qed
\end{proof}

 \section{The main results}
Here we prove that the problem is decidable for two
families of languages, namely the linear context-free languages and the 
linear bounded languages.

\subsection{The bounded semilinear languages}

We solve the easier case. We recall that a  \emph{bounded semilinear}  language
is a finite union of \emph{linear} languages which are languages of  the form
\begin{equation}
\label{eq:linear}
L=\{w_{1}^{n_{1}}\cdots w_{k}^{n_{k}} \mid(n_{1},\ldots, n_{k}) \in R\}
\end{equation}
for some fixed words $w_{i}\in \Sigma^*$ for $i=1, \ldots, k$ and $R\subseteq \N^k$
is a linear set, i.e., there exists $v_{0}, v_{1}, \ldots, v_{p}\in \N^k$
such that 
$$
R=\{v_{0}+ \lambda_{1} v_{1} + \cdots + \lambda_{p} v_{p}  \mid \lambda_{1},  \ldots,  \lambda_{p}\in \N\}
$$

\begin{proposition}\label{propsemilinear}
If $L$ is  bounded semilinear then  
its closure $\Closure{\varphi(L)}$ is semialgebraic.
Furthermore, if the quantum automaton ${\cal Q}$ is rational, 
the $(L, Q)$ intersection %$L\subseteq |{\cal Q}_{\leq} |$ 
is decidable.
\end{proposition}

 \begin{proof}
 Because the semialgebraic sets are closed under finite union,
it suffices to consider 
the case where the language is of the form (\ref{eq:linear}).
For $t=0, \ldots, p$ set $v^T_{t}=(v_{t,1}, \ldots, v_{t,k})$  and
consider the orthogonal matrices 
$$
g_{t}=\left ( 
 \begin{array}{cccc}
  \varphi (w_{1})^{v_{t,1}}  & 0 & 0 & 0 \\
0 &    \varphi  (w_{2})^{v_{t,2}}  & 0 & 0 \\
\vdots &     \vdots & \vdots  &    \vdots\\
0  &    0  &   0 &  \varphi  (w_{k})^{v_{t,k}}
\end{array}
\right ) 
$$
Set $G=\{g_{i}\mid i=1, \ldots, p\}$. In virtue of 
Theorem \ref{th:semigroup-in-compact} and Proposition 
\ref{pr:finitely-generated}  the set
$\Closure{G^*}$ is semialgebraic
and it is effectively eventually definable if the coefficients 
of the quantum automata are rational.
By Corollary \ref{cor:closure-of-product} we have 
$$
\Closure{g_{0}G^*}
= g_{0}\Closure{G^*}
$$
and by Proposition \ref{pr:union-prod-phi} 
this product is semialgebraic (resp. and  
effectively eventually definable if the coefficents 
of the quantum automaton are rational).
By  Proposition \ref{pr:product-of-blocks},
 $\Closure{\varphi(L)}$ is semialgebraic (resp. and  
effectively eventually definable if the coefficients 
of the quantum automaton are rational).
In the latter case the $(L, Q)$ intersection is decidable by Proposition 
\ref{pr:meta-resuositionlt} which completes the proof.
\qed\end{proof}

 \subsection{The case of context-free  languages}
\label{subsec:cf-lin}

Here we show that $\closure(\varphi(L))$ is effectively eventually
definable for languages generated by linear grammars
and rational quantum automata.

We adopt the notation from Section \ref{sec:context-free} for   context-free grammars.
We recall the following notion that will be used in the proof of the next result (see \cite{js}). A subset of a monoid $M$ is
\emph{regular} if it is recognized by some finite $M$-automaton
which differs from an ordinary finite 
nondeterministic automaton
over the free monoid by the fact the 
transitions are labeled by elements in $M$.

\begin{proposition}
If $L$ is generated by a context-free grammar, then 
$\Closure{\varphi(L)}$ is semialgebraic. Furthermore,  
if the grammar is linear and  if the quantum automaton is rational
then $\Closure{\varphi(L)}$  is effectively 
eventually  definable and the $(L, Q)$ intersection
is decidable.
\end{proposition}

\begin{proof}
With the notation of Section
\ref{sec:context-free} the language $L$ is a finite
union of languages of the form
$C_{S}\diamond L'' $ with
\begin{equation}
\label{eqfla:claim2-a} 
L''= 
w_{1}L_{A_{1}}w_{2}L_{A_{2}} \cdots w_{\ell}L_{A_{\ell}} w_{\ell+1}
\end{equation}
where, for every  $1\leq i \leq \ell +1$, 
 $w_{i}\in \Sigma^*$ and $A_{i}\in V'$.
It suffices to show by induction on the number
of nonterminal symbols that, with the previous
notation, the subsets  
\begin{equation}
\label{eq:L-double-prime}
 \Closure{\varphi(C_{S}\diamond L'')} 
\end{equation}
are semialgebraic in all cases and  effectively eventually definable when the quantum automaton is rational and the grammar of the language is linear.
As a preliminary remark let us show this property for 
$\Closure{\varphi(C_{S})}$.
Define $\varphi^T: \Sigma^*\rightarrow \R^{n\times n}$
as $\varphi^T(u)=\varphi(u)^T$ and set 
 $$M=\{\varphi(a)\oplus \varphi^T(b)\mid (a,b)\in C_{S}\}.
 $$
%
%where, applied to a matrix, the superscript $T$ represents its transpose.

Observe that $M$ is a monoid since
if $\varphi(a)\oplus \varphi^T(b)$ and $\varphi(c)\oplus \varphi^T(d)$ are 
in $M$ then we have
$$
\varphi^T(b) \varphi^T(d)=\varphi(b)^T \varphi(d)^T
=(\varphi(d) \varphi(b))^T= \varphi(d b)^T= \varphi^T(d b)
$$
which yields 
$$
\begin{array}{l}
(\varphi(a)\oplus \varphi^T(b)) (\varphi(c)\oplus \varphi^T(d))
=\varphi(a c)\oplus \varphi^T(db).
\end{array}
$$
As a first consequence, 
by Corollary \ref{cor:L-star-is-algebraic}, 
$\Closure{M}$ is algebraic.
Furthermore  we can show that $\Closure{M}$ is effectively eventually definable. Indeed $M$ is a regular submonoid
of  the group of orthogonal matrices $O_{n}\oplus  O_{n}$
if the grammar is linear.
Precisely, it is recognized by the finite 
$O_{2n}$-automaton whose states
are the nonterminal symbols, the transitions
are of the form $A\trans{\varphi(a)\oplus  \varphi^T(b)} B$ where 
$A\rightarrow aBb$ is a rule of the grammar
and where the initial and final states coincide 
with $S$.  Now, the subgroup generated 
by a regular subset of a monoid has an 
effective finite generating set  \cite{Anisimov} (see also \cite{js})  
and thus by  Proposition \ref{pr:finitely-generated}
$\Closure{M}$ is effectively eventually definable
if $\varphi(\Sigma^*)\subseteq \Q^{n\times n}$.

  We now proceed with the proof by induction on the number of nonterminal symbols.
If  the set of nonterminal symbols is reduced to $S$ 
then $L$ is reduced to $C_S\diamond L'(G)$ and $L'(G)$ is finite. 
We may further assume that there is a unique
terminal rule $S\rightarrow w$.
By Theorem \ref{th:fundamental} we have
$$
\Closure{\varphi(L)}
=\{X\varphi(w)Y^T\mid X\oplus Y\oplus \{\varphi(w)\}\in \Closure{M\oplus \varphi(w)}\}
$$
By Corollary \ref{cor:closure-of-product} we have
$$\Closure{M\oplus \varphi(w)}
= \Closure{M}\oplus \Closure{\varphi(w)}
= \Closure{M}\oplus \varphi(w)
$$
which, by Proposition \ref{pr:blocks-of-product},
is semialgebraic, resp. effectively eventually definable.
In that latter case the $(L, Q)$ intersection is decidable.

Now assume $V$ contains more than one nonterminal symbol.
We first prove that for each  nonterminal symbol $A$, 
$
\Closure{\varphi(C_{S}\diamond L_{A})}
$ 
is  semialgebraic in the general case 
and effectively eventually definable when the grammar is 
linear and the quantum automaton is rational.
By Theorem \ref{th:fundamental} 
and  Corollary \ref{cor:closure-of-product},
 $\Closure{\varphi(C_{S}\diamond L'')}$
is the subset 
$$
\Closure{\varphi(C_{S}\diamond L'')}
=\{XZY^T\mid X\oplus Y\oplus Z\in \Closure{M}\oplus \Closure{\varphi(L''}\}
$$
 with $L''$ as in  (\ref{eqfla:claim2-a}),
 i.e.,
$$
\{XZY^T\mid X\oplus Y\oplus Z \in \Closure{M}\oplus \Closure{\varphi(w_{1})\varphi(L_{A_{1}})\cdots 
\varphi(w_{\ell})\varphi(L_{A_{\ell}})\varphi(w_{\ell+1})}\}
$$
 By Cororally  \ref{cor:closure-of-product}
we have
$$
\begin{array}{ll}
& \Closure{\varphi(w_{1})\varphi(L_{A_{1}})\cdots 
\varphi(w_{\ell})\varphi(L_{A_{\ell}})\varphi(w_{\ell+1}))\\
= &\varphi(w_{1})\Closure{\varphi(L_{A_{1}})}\cdots 
\varphi(w_{\ell})\Closure{\varphi(L_{A_{\ell}})}\varphi(w_{\ell+1})}
\end{array}
$$
which shows, via Proposition \ref{pr:union-prod-phi}
and by induction hypothesis 
that this subset is semialgebraic, resp. effectively,
eventually definable. Then  its direct sum with
$\Closure{M}$ is semialgebraic and
effectively,
eventually definable if the grammar is
linear and the quantum automaton is rational. We conclude by applying Proposition
\ref{pr:product-of-blocks}.  \qed
  \end{proof}

  \section{Complement of context-free languages}\label{al}
In this section we prove that there is a language $L$ such that $(i)$ the 
complement of $L$ is  context-free and $ (ii)$  
$\Closure{\varphi (L)}$  is not semialgebraic.

 Given a binary representation of a real $\alpha=0.b_{1}\cdots b_{n} \cdots $, we define its \emph{approximation 
sequence} $(\alpha[k])_{k\geq 0}$ as the sequence of its successive truncations   $\alpha[k]=0.b_{1}\ldots b_{k}$.

\begin{lemma}
Let  $0<\alpha<1$ be an irrational. There exist infinitely many 
rationals $\frac{q}{n}$ such that 
$$
\left |\alpha[1+2 \ell(n)]-\frac{q}{n}\right |<\frac{1}{n^2}
$$
holds, where $\ell(n)=\lfloor \log_{2} n\rfloor$.
\end{lemma}

\begin{proof}
By the triangular inequality we have 
$$
\left |\alpha[1+2 \ell(n)]-\frac{q}{n}\right | \leq \left |\alpha[1+2 \ell(n)]-\alpha \right | + \left |\alpha -\frac{q}{n}\right |
$$
By  the definition of the approximation sequence we get
$$
|\alpha[1+2 \ell(n)]-\alpha|
\leq \frac{1}{2^{1+2 \ell(n)}} = \frac{1}{2} \times  \frac{1}{2^{2 \lfloor \log_{2} n\rfloor}}
\leq \frac{1}{2} \times \frac{1}{n^{2}}
$$
Now by Hurwitz Theorem there exist infinitely many rationals $\frac{p}{n}$
for which 
$$
 \left |\alpha -\frac{p}{n}\right |\leq \frac{1}{\sqrt{5}} \times \frac{1}{n^{2}}
$$
We conclude by combining these last two inequalities.\qed
\end{proof}

We now fix an irrational $0<\alpha<1$. Consider the
orthogonal matrix 
$$
M_{\alpha}=
\left(
\begin{array}{ll}
\cos 2\pi \alpha & \sin 2\pi \alpha \\
-\sin  2\pi \alpha & \cos 2\pi \alpha
\end{array}
\right)
$$
and  the morphism $\varphi_{\alpha}: b^*\rightarrow O_{n}$ from the free monoid generated by the letter $b$ and the
group $O_n$
defined by $\varphi_{\alpha}(b)=M_{\alpha}$. Furthermore  set 
\begin{equation}
\label{eq:L-alpha}
L(\alpha)=\left \{b^n\mid \exists \  q\in \N \mid   \left |\alpha[1+2\ell(n)]-\frac{q}{n}\right |< \frac{1}{n^2}\right \}
\end{equation}

\begin{lemma}
\label{le:non-semialgebraic}
  If $0<\alpha<1$ is an irrational,
  the topological closure   $\Closure{\varphi (L(\alpha)}$ 
  is not semialgebraic.
\end{lemma} 
\begin{proof}
Since the element in position $(1,1)$ of the matrix
$\varphi_{\alpha} (b^n)$ is $\cos 2\pi n \alpha$ and since the projection 
of a semialgebraic set is semialgebraic, it suffices 
to show that 
$$
 \Closure{\{ \cos 2\pi n \alpha \mid b^n\in L(\alpha)\}}
$$
is not semialgebraic.

Observe that $n\not=n'$ implies $\cos 2\pi n \alpha\not= \cos 2\pi n' \alpha$
since $\alpha$ is irrational. In particular, the set
$\{ \cos 2\pi n \alpha \mid b^n\in L(\alpha)\}$ is infinite. 

Now we verify that $1$ is the unique limit point. Indeed, by definition
$b^n\in L(\alpha)$ implies that for some
integer $q$ we have
$$
\left |\alpha[1 + 2\ell(n)]-\frac{q}{n}\right |\leq \frac{1}{n^2}
$$
For such an integer $q$ we have
$$
\left |\alpha -\frac{q}{n}\right |  \leq \left |\alpha -\alpha[1 + 2\ell(n)]\right |+\left |\alpha[1 + 2\ell(n)]- \frac{q}{n} \right |\\
   \leq  \frac{1}{2n^2}  + \frac{1}{n^2} =  \frac{3}{2}\times \frac{1}{n^2}
$$
 Consequently, $|n\alpha -q |\leq \frac{3}{2}\times  \frac{1}{n}$. Now we compute
$$
1\geq \cos 2\pi n\alpha = \cos 2\pi (n\alpha -q)=\cos 2\pi |n\alpha -q |\geq \cos  \frac{3\pi}{n}
$$
which proves that the closure $ \Closure{\{ \cos 2\pi n \alpha \mid b^n\in L(\alpha)\}}$  consists of a unique 
limit point and of infinitely many isolated points. This is not a semialgebraic set 
since the semialgebraic sets on the reals are finite unions of intervals.\qed
\end{proof} 

Now we state the main result of this section.

\begin{theorem}\label{alberto-christian}
  There is a  language $ L \subseteq   \Sigma^* $ and a morphism  $\varphi:\Sigma^*\rightarrow \Q^{n\times n}$, 
  assigning an orthogonal matrix to every word of $\Sigma^*$,   such that (i) $ L$  is the complement of a context-free language  (ii) 
  the topological closure   $\Closure{\varphi (L)}$ is not semialgebraic.
\end{theorem}

\begin{proof} 
  Consider a one tape Turing machine implementing the 
following procedure for recognizing the language $L(\alpha)$ defined
in (\ref{eq:L-alpha}):

\medskip

\noindent Input $b^n$\\
$A\leftarrow  \alpha[1+2\ell(n)]$\\
$F\leftarrow 0$\\
for $q=1$ to $n$, if $|A-\frac{q}{n}|<\frac{1}{n^2}$ then $F\leftarrow 1$\\
if $F=1$ then write $ab^n$, \\
position the head on the rightmost occurrence
of $b$, \\
change to a new state $\hat{q}$, move the reading head to the leftmost 
cell while staying in state $\hat{q}$  \\
stop when reaching the occurrence $a$.

\bigskip We know that the  computation histories of a Turing machine, 
i.e., the set of sequences of configurations properly separated by a new symbol 
is, as a language,  the intersection of two linear context-free languages
(see, for instance, \cite{Hu}, Lemma 8.6).
Let $\Hist(b^n)$ be the history associated to the input $b^n$. 
Let $\Gamma$ be the disjoint union of the symbols comprising the 
input and tape alphabets along with the states including the
special state $\hat{q}$. 
With $\alpha= \arctan \frac{3}{4}$ we get the orthogonal matrix
$$
M_{\alpha}=
\left(
\begin{array}{cc}
\frac{3}{5}  & \frac{4}{5} \\
-\frac{4}{5} & \frac{3}{5}
\end{array}
\right)
$$
Define the morphism $\varphi:\Gamma^*\rightarrow O_{2}$ by 
$$
\varphi(c): =
\left\{
\begin{array}{rl}
I & \mbox{ if } c\in \Gamma\setminus \{\hat{q}\} \\
M_{\alpha}  &  \mbox{ if } c=\hat{q}
\end{array}
\right.
$$
By applying the result mentioned above to the Turing machine implementing the procedure for 
recognizing the language $L(\alpha)$, we have that there  exist two linear context-free languages $L_ 1$ and $ L_ 2 $
such that 
$$
\Hist(L_{\alpha}) = L_ 1\cap L_ 2 = (L_ 1^{c} \cup L_ 2^{c})^{c} 
 $$ 
Since $L_ {1}$ and $ L_ {2}$  are linear and deterministic context-free,
$L_ {1}^{c} \cup L_ {2}^{c}$ is (not necessarily  deterministic)
context-free and thus $\Hist(L_{\alpha})$  is the complement of a  context-free language. 
But then 
$$
\varphi(\Hist(L_{\alpha}))=\{M_{\alpha}^n\mid b^n\in L_{\alpha}\}
$$
We conclude by applying Lemma \ref{le:non-semialgebraic}.
\qed \end{proof}

\end{document}